\documentclass[12pt]{article}

\usepackage{amsmath}
\usepackage{amssymb}
\usepackage{amsthm}
\usepackage{hyperref}
\usepackage{enumerate}
\usepackage{bbm}
\usepackage{color}
\newtheorem{thm}{Theorem}[section]
\newtheorem{co}[thm]{Corollary}
\newtheorem{lem}[thm]{Lemma}

\newtheorem{assumption}[thm]{Assumption}

\newtheorem{definition}[thm]{Definition}

\newtheorem{example}[thm]{Example}

\newtheorem{remark}[thm]{Remark}
\newenvironment{rem}{\begin{remark}\rm}{\end{remark}}

\newcommand{\E}{\mathbb{E}}
\newcommand{\EX}{{\mathbb{E}}}
\newcommand{\PX}{{\mathbb{P}}}

\title{On Sampling Continuous-Time AWGN Channels}

\author{\small \begin{tabular}{ccc}
Guangyue Han & Shlomo Shamai\\
The University of Hong Kong & Technion-Israel Institute of Technology\\
email:  ghan@hku.hk & email: sshlomo@ee.technion.ac.il\\
\end{tabular}}

\date{\today}

\begin{document} \maketitle

\begin{abstract}
For a continuous-time additive white Gaussian noise (AWGN) channel with possible feedback, it has been shown that as sampling gets infinitesimally fine, the mutual information of the associative discrete-time channels converges to that of the original continuous-time channel. We give in this paper more quantitative strengthenings of this result, which, among other implications, characterize how over-sampling approaches the true mutual information of a continuous-time Gaussian channel with bandwidth limit. The assumptions in our results are relatively mild. In particular, for the non-feedback case, compared to the Shannon-Nyquist sampling theorem, a widely used tool to connect continuous-time Gaussian channels to their discrete-time counterparts that requires the band-limitedness of the channel input, our results only require some integrability conditions on the power spectral density function of the input.
\end{abstract}

\section{Introduction}

In this paper, we are concerned with the following continuous-time AWGN channel:
\begin{equation} \label{ct}
Y(t) = \int_0^t X(s) ds + B(t),~~~t \geq 0,
\end{equation}
where $\{B(t)\}$ denotes the standard Brownian motion. We will examine the channel (\ref{ct}) for both the non-feedback and feedback cases. More specifically, the non-feedback case refers to the scenario when the feedback is not allowed in the channel, and thereby the channel input $\{X(s)\}$ takes the following form:
\begin{equation} \label{non-feedback-case}
X(s) = g(s, M),
\end{equation}
where $g$ is a real-valued deterministic function, $M$ is a random variable, independent of $\{B(t)\}$ and often interpreted as the {\it message} to be transmitted through the channel. By contrast, for the feedback case, $X(s)$ may also depend on the previous output with the following form:
\begin{equation} \label{feedback-case}
X(s) = g(s, M, Y_0^s),
\end{equation}
where $Y_0^s \triangleq \{Y(r): 0 \leq r \leq s\}$ is the channel output up to time $s$ that is fed back to the sender, which will be referred to as the {\em channel feedback} up to time $s$. Here, we remark that for the non-feedback case as in (\ref{non-feedback-case}), it can be verified that for any $T > 0$,
$$
I(X_0^T; Y_0^T) = I(M; Y_0^T),
$$
which however does not hold true for the feedback case as in (\ref{feedback-case}) when the channel is intricately characterized by the following stochastic differential equation:
\begin{equation} \label{SDE}
Y(t) = \int_0^t g(s, M, Y_0^s) ds + B(t),~~~t \geq 0,
\end{equation}
rather than a simple input-output equation.

For any $T > 0$, we say that $t_{0}, t_{1}, \dots, t_{n} \in \mathbb{R}$~\footnote{Here and hereafter, all $t_i$ depend on $T$ and $n$, however we suppress this notational dependence for brevity.} are {\em evenly spaced} over $[0, T]$ if $t_{0} = 0$, $t_{n} = T$ and $\delta_{T, n} \triangleq t_{i} - t_{i-1} = T/n$ for all feasible $i$. Sampling the continuous-time channel (\ref{ct}) over the time window $[0, T]$ with respect to evenly spaced $t_{0}, t_{1}, \dots, t_{n}$, we obtain the following discrete-time Gaussian channel~\footnote{The sampler associated with (\ref{dt}) has been examined from a communication system design perspective and termed as the {\em integrate-and-dump} filter; see, e.g.,~\cite{CarlsonCrilly2009} and references therein.}
\begin{equation} \label{dt}
Y(t_{i}) = \int_{0}^{t_{i}} X(s) ds + B(t_{i}), \quad i = 1, 2, \dots, n.
\end{equation}
It turns out that if the sampling is ``fine'' enough, the mutual information of the continuous-time Gaussian channel (\ref{ct}) over $[0, T]$ can be ``well-approximated'' by that of the discrete-time Gaussian channel (\ref{dt}). More precisely, it has been established in~\cite{LiuHan2019} that under some mild assumptions,
\begin{equation} \label{limit-version}
\lim_{n \to \infty} I(M; Y(\Delta_{T, n})) = I(M; Y_0^T),
\end{equation}
where
$$
\Delta_{T, n} \triangleq \{t_{0}, t_{1}, \dots, t_{n}\}, \quad Y(\Delta_{T, n}) \triangleq \{Y(t_{0}), Y(t_{1}), \dots, Y(t_{n})\}.
$$
This result connects a continuous-time Gaussian channels with its associative discrete-time versions, which has been used to recover a classical information-theoretic formula~\cite{FourPersons}.

Strictly speaking, the above result is not new; as a matter of fact, it has been ``known'' for decades. Indeed, though not explicitly stated, a more general result, which implies (\ref{limit-version}) as a special corollary, follows from an appropriately modified argument in~\cite{GelfandYaglom1957} (see more details in Section~\ref{already-by-GY}). Moreover, the functional analysis approach as in~\cite{GelfandYaglom1957} tackles  more general channels and thus reveals the essence beneath the connection as in (\ref{limit-version}). By comparison, when it comes to merely establishing the result (\ref{limit-version}), our stochastic approach in~\cite{LiuHan2019} is indirect and cumbersome, only scratching the surface of the connection.

That being said, the stochastic calculus approach does allow us to capitalize on the peculiar characteristics of the continuous-time AWGN channel formulated as in (\ref{ct}), which is already evidenced by the approximation theorems established therein that do not seem to follow from the aforementioned functional analysis approach. The present paper will continue to employ the stochastic calculus approach to conduct a closer examination of (\ref{ct}) and quantitatively strengthen (\ref{limit-version}), particularly by zooming in on its convergence behavior. Our results encompass both the non-feedback case (Theorem~\ref{non-feedback-theorem}) and the feedback case (Theorem~\ref{feedback-theorem}), which may be used for a finer analysis of the channel (\ref{ct}) with possible feedback from an information-theoretic perspective. In particular, among other possible implications, our results characterize how over-sampling approaches the true mutual information of a continuous-time band-limited AWGN channel over a finite time window (Corollary~\ref{non-feedback-corollary}).

We would like to add that the assumptions imposed in our results are rather mild. Indeed, the celebrated Shannon-Nyquist sampling theorem~\cite{sh49, ny24}, a widely used tool to connect a continuous-time non-feedback AWGN channel to its discrete-time versions, requires the band-limitedness of the channel input. By comparison, with considerably stronger conclusions, Theorem~\ref{non-feedback-theorem} only require some integrability conditions on the power spectral density function of the channel input and Theorem~\ref{feedback-theorem} only requires some mild regularity conditions that are more or less standard in the theory of stochastic differential equations.

\section{Notation and Preliminaries} \label{notations}

We use $(\Omega,\mathcal{F},\PX)$ to denote the underlying probability space, and $\EX$ to denote the expectation with respect to the probability measure $\PX$. As is typical in the theory of SDEs, we assume the probability space is equipped with a filtration $\{\mathcal{F}_t: 0 \leq t < \infty\}$, which satisfies the {\em usual conditions}~\cite{ka91} and is rich enough to accommodate the standard Brownian motion $\{B(t): 0 \leq t < \infty\}$. Throughout the paper, we will use uppercase letters (e.g., $X$, $Y$, $Y^{(n)}$) to denote random variables or processes, and their lowercase counterparts (e.g., $x$, $y$, $y^{(n)}$) to denote their realizations.

Let $C[0, \infty)$ denote the space of all continuous functions over $[0, \infty)$, and for any $t > 0$, let $C[0, t]$ denote the space of all continuous functions over $[0, t]$. As usual, we will equip the space $C[0, \infty)$ with the filtration $\{\mathcal{B}_t\}_{0 \leq t < \infty}$, where $\mathcal{B}_{\infty}$ denotes the standard Borel $\sigma$-algebra on the space $C[0, \infty)$ and $\mathcal{B}_t = \pi_t^{-1}(\mathcal{B}_{\infty})$, where $\pi_t : C[0, \infty) \to C[0, t]$ is defined as $(\pi_tx)(s) = x(t\wedge s)$.

For any $\varphi \in C[0, \infty)$, we use $\varphi(\{t_1, t_2, \dots, t_m\})$ to denote $\{\varphi(t_1), \varphi(t_2), \dots, \varphi(t_n)\}$ and $\varphi_0^t$ to denote $\{\varphi(s): 0 \leq s \leq t\}$. The sup-norm of $\varphi_0^t$, denoted by $\|\varphi_0^t\|$, is defined as $\|\varphi_0^t\| \triangleq \sup_{0 \leq s \leq t} |\varphi(s)|$; and similarly, we define $\|\varphi_0^t-\psi_0^t\| \triangleq \sup_{0 \leq s \leq t} |\varphi(s)-\psi(s)|$. For any $\varphi, \psi \in C[0, \infty)$, slightly abusing the notation, we define $\|\varphi_0^s-\psi_0^t\| \triangleq \|\hat{\varphi}_0^{\infty}-\hat{\psi}_0^{\infty}\|$, where $\hat{\varphi}, \hat{\psi} \in C[0, \infty)$ are ``stopped'' versions of $\varphi, \psi$ at time $s, t$, respectively, with $\hat{\varphi}(r) \triangleq \varphi(r \wedge s)$ and $\hat{\psi}(r) \triangleq \psi(r \wedge t)$ for any $r \geq 0$.

Let $X, Y, Z$ be random variables defined on the probability space $(\Omega,\mathcal{F},\PX)$, which will be used to illustrate most of the notions and facts in this section (note that the same notations may have different connotations in other sections). Note that in this paper, a random variable can be discrete-valued with a probability mass function, real-valued with a probability density function or path-valued (more precisely, $C[0, \infty)$ or $C[0, t]$-valued).

For any two probability measures $\mu$ and $\nu$, we write $\mu\sim\nu$ to mean they are equivalent, namely, $\mu$ is absolutely continuous with respect to $\nu$ and vice versa. For any two path-valued random variables $X_0^t = \{X(s); 0 \leq s \leq t\}$ and $Y_0^t =\{Y(s); 0 \leq s \leq t\}$, we use $\mu_{X_0^t}$ and $\mu_{Y_0^t}$ to denote the probability distributions on $\mathcal{B}_t$ induced by $X_0^t$ and $Y_0^t$, respectively; and if $\mu_{Y_0^t}$ is absolutely continuous with respect to $\mu_{X_0^t}$, we write the Radon-Nikodym derivative of $\mu_{Y_0^t}$ with respect to $\mu_{X_0^t}$ as $d\mu_{Y_0^t}/d\mu_{X_0^t}$. We use $\mu_{Y_0^t|Z=z}$ denote the probability distribution on $\mathcal{B}_t$ induced by $Y_0^t$ given $Z=z$, and $d\mu_{Y_0^t|Z=z}/d\mu_{X_0^t|Z=z}$ to denote the Radon-Nikodym derivative of $Y_0^t$ with respect to $X_0^t$ given $Z=z$. Obviously, when $Z$ is independent of $X$, $d\mu_{Y_0^t|Z=z}/d\mu_{X_0^t|Z=z}= d\mu_{Y_0^t|Z=z}/d\mu_{X_0^t}$.

By definition, for $\EX[X|\sigma(Y, Z)]$, the conditional expectation of $X$ with respect to the $\sigma$-algebra generated by $Y$ and $Z$, there exists a $\sigma(Y) \otimes \sigma(Z)$-measurable function $\Psi(\cdot, \cdot)$ such that $\Psi(Y, Z)=\EX[X|\sigma(Y, Z)]$.
For notational convenience, we will in this paper simply write $\EX[X|\sigma(Y, Z)]$ as $\EX[X|Y, Z]$, and $\Psi(y, z)$ as $\EX[X|y, z]$ and furthermore, $\Psi(Y, z)$ as $\EX[X|Y, z]$.

A {\em partition} of the probability space $(\Omega,\mathcal{F},\PX)$ is a disjoint collection of elements of $\mathcal{F}$ whose union is $\Omega$. It is well known there is a one-to-one correspondence between finite partitions and finite sub-$\sigma$-algebras of $\mathcal{F}$. For a finite sub-$\sigma$-algebra $\mathcal{H} \subset \mathcal{F}$, let $\eta(\mathcal{H})$ denote the corresponding finite partition. The entropy of a finite partition $\xi=\{A_1, A_2, \cdots, A_m\}$, denoted by $H(\xi)$, is defined as $H(\xi) \triangleq \sum_{i=1}^m -\PX(A_i)\log \PX(A_i)$, whereas the conditional entropy of $\xi$ given another finite partition $\zeta=\{B_1, B_2, \dots, B_n\}$, denoted by $H(\xi|\zeta)$, is defined as $H(\xi|\zeta) \triangleq \sum_{j=1}^n \sum_{i=1}^m -\PX(A_i \cap B_j) \log \PX(A_i|B_j)$. The mutual information between the above-mentioned two partitions $\xi$ and $\zeta$, denoted by $I(\xi; \zeta)$, is defined as $I(\xi; \zeta) \triangleq \sum_{j=1}^n \sum_{i=1}^m -\PX(A_i \cap B_j) \log \PX(A_i \cap B_j)/\PX(A_i) \PX(B_j)$.

For the random variable $X$, we define
$$
\eta(X) \triangleq \{\eta(\mathcal{H}): \mathcal{H} \mbox{ is a finite sub-}\sigma\mbox{-algebra of } \sigma(X)\}.
$$
The {\em entropy} of the random variable $X$, denoted by $H(X)$, is defined as
$$
H(X) \triangleq \sup_{\xi \in \eta(X)} H(\xi).
$$
The {\em conditional entropy} of $Y$ given $X$, denoted by $H(Y|X)$, is defined as
$$
H(Y|X) \triangleq \inf_{\xi \in \eta(X)} \sup_{\zeta \in \eta(Y)} H(\zeta|\xi).
$$
Here, we note that if $X$ and $Y$ are independent, then obviously it holds that
\begin{equation} \label{Y-X-Y}
H(Y|X)=H(Y).
\end{equation}
The {\em mutual information} between $X$ and $Y$, denoted by $I(X; Y)$, is defined as
$$
I(X; Y) \triangleq \sup_{\xi \in \eta(X), \; \zeta \in \eta(Y)} I(\xi; \zeta).
$$
A couple of properties of mutual information are in order. First, it can be shown, via a concavity argument, that the mutual information is always non-negative. Second, the mutual information is determined by the $\sigma$-algebras generated by the corresponding random variables. For example, for any random variables $X', Y', X'', Y''$,
\begin{equation} \label{I=I}
I(X'; Y')=I(X'; Y'') \mbox{ if } \sigma(X') = \sigma(X'') \mbox{ and } \sigma(Y')=\sigma(Y'')
\end{equation}
and
\begin{equation} \label{I<=I}
I(X'; Y') \leq I(X'; Y'') \mbox{ if } \sigma(X') \subset \sigma(X'') \mbox{ and }  \sigma(Y') \subset \sigma(Y'').
\end{equation}
For another example, we have
$$
I(X; Y)=I(X, X; Y, Y+X),~~~I(X; Y) \leq I(X; Y, Z).
$$

It turns out that for the case that $X, Y, Z$ are all discrete random variables, all the above-mentioned notions are well-defined and can be computed rather explicitly: $H(X)$ can be computed as $H(X) = \EX[-\log p_X(X)]$, where $p_X(\cdot)$ denotes the probability mass function of $X$; $H(Y|X)$ can be computed as $H(Y|X) = \EX[-\log p_{Y|X}(Y|X)]$, where $p_{Y|X}(\cdot|\cdot)$ denotes the conditional probability mass function of $Y$ given $X$; $I(X; Y)$ can be computed as
\begin{equation} \label{mutual-information}
I(X; Y) = \EX\left[\log \frac{p_{Y|X}(Y|X)}{p_Y(Y)}\right].
\end{equation}
The mutual information is intimately related to entropy. As an example, one verifies that
\begin{equation} \label{I-H}
I(X; Y)=H(Y)-H(Y|X).
\end{equation}
Note that the quality (\ref{I-H}) may fail if non-discrete random variables are involved, since the corresponding entropies $H(Y)$ and $H(Y|X)$ can be infinity. For the case of real-valued random variables with density, this issue can be circumvented using the notion of differential entropy, as elaborated below.

Now, assume that $Y$ is real-valued with probability density function $f_Y(\cdot)$. The {\em differential entropy} of $Y$, denoted by $h(Y)$, is defined as $h(Y) \triangleq \EX[-\log f_Y(Y)]$. And the {\em differential conditional entropy} of $Y$ given a finite partition $\xi$, denoted by $h(Y|\zeta)$, is defined as $h(Y|\zeta) \triangleq \sum_{j=1}^n \PX(A_i) \int f_{Y|A_i}(x) \log f_{Y|A_i}(x) dx$. The {\em differential conditional entropy} of $Y$ given $X$ (which can be discrete-, real- or path-valued), denoted by $h(Y|X)$, is defined as $h(Y|X) \triangleq \inf_{\xi \in \eta(X)} h(Y|\xi)$ (which, similarly as in (\ref{Y-X-Y}), reduces to $h(Y)$ if $Y$ is independent of $X$); in particular, if the conditional probability density function $f_{Y|X}(\cdot|\cdot)$ exists, then $h(Y|X)$ can be explicitly computed as $\EX[-\log f_{Y|X}(Y|X)]$. As mentioned before, the aforementioned failure of (\ref{I-H}) can be salvaged with the notion of differential entropy:
\begin{equation} \label{I-h}
I(X; Y)=h(Y)-h(Y|X).
\end{equation}

Here we emphasize that all the above-mentioned definitions naturally carry over to the setting when some/all of the invovled random variables are replaced by vectors of random variables. For a quick example, let $Y = \{Y_1, Y_2, \dots, Y_n\}$, where each $Y_i$ is a real-valued random variable with density. Then, the differential entropy $h(Y)$ of $Y$ is defined as
$$
h(Y)=h(Y_1, Y_2, \dots, Y_n) \triangleq \EX[- \log f_{Y_1, Y_2, \dots, Y_n}(Y_1, Y_2, \dots, Y_n)],
$$
where $f_{Y_1, Y_2, \dots, Y_n}$ is the joint probability density function of $Y_1, Y_2, \dots, Y_n$.

The notion of mutual information can be further extended to generalized random processes, which we will only briefly describe and we refer the reader to~\cite{GelfandYaglom1957} for a more comprehensive exposition.

The mutual information between two generalized random processes $X=\{X(t)\}$ and $Y=\{Y(t)\}$ is defined as
\begin{equation} \label{general-definition}
I(X; Y) = \sup I(X(\phi_1), X(\phi_2), \dots, X(\phi_m); Y(\psi_1), Y(\psi_2), \dots, Y(\psi_n)),
\end{equation}
where the supremum is over all possible $n, m \in \mathbb{N}$ and all possible testing functions $\phi_1, \phi_2, \dots, \phi_m$ and $\psi_1, \psi_2, \dots, \psi_n$, and we have defined
$$
X(\phi_i) = \int X(t) \phi_i(t) dt, \quad i=1, 2, \dots, m, \mbox{ and } Y(\psi_j) = \int Y(t) \psi_j(t) dt, \quad j=1, 2, \dots, n.
$$
It can be verified that the general definition of mutual information as in (\ref{general-definition}) includes all previous definitions as special cases; moreover, when one of $X$ and $Y$, say, $Y$, is a random variable, the general definition boils down to
$$
I(X; Y) = \sup I(X(\phi_1), X(\phi_2), \dots, X(\phi_m); Y),
$$
where the supremum is over all possible $n \in \mathbb{N}$ and all possible testing functions $\phi_1, \phi_2, \dots, \phi_m$.

For the channel (\ref{ct}) with the input as in (\ref{non-feedback-case}) or (\ref{feedback-case}) , it is known that its mutual information over $[0, T]$ can be computed as (see, e.g., ~\cite{pi64,ih93}):
\begin{equation} \label{definition-mutual-information}
I(M; Y_0^T)=\begin{cases}
\EX\left[\log \frac{d \mu_{M, Y_0^T}}{d \mu_{M} \times \mu_{Y_0^T}}(M, Y_0^T)\right], &\mbox{ if } \frac{d \mu_{M, {Y_0^T}}}{d \mu_{M} \times \mu_{Y_0^T}} \mbox{ exists },\\
\infty, &\mbox{ otherwise },
\end{cases}
\end{equation}
where $d \mu_{M, Y_0^T}/d \mu_M \times \mu_{Y_0^T}$ denotes the Radon-Nikodym derivative of $\mu_{M, Y_0^T}$ with respect to $\mu_M \times \mu_{Y_0^T}$.

\section{The Non-Feedback Case}

In this section, we examine the AWGN channel (\ref{ct}) for the non-feedback case and give quantitative strengthenings of (\ref{limit-version}), detailed in the following theorem.
\begin{thm} \label{non-feedback-theorem}
For the continuous-time AWGN channel (\ref{ct}), suppose that the channel input $\{X(t)\}$ is a stationary stochastic process with power spectral density function $f(\cdot)$~\footnote{More precisely, the channel input $\{X(t): t \geq 0\}$ can be extended to a bi-infinite stationary stochastic process $\{X(t): -\infty < t < \infty\}$ with power spectral density function $f(\cdot)$.}. Then, the following two statements hold:
\begin{enumerate}
\item[(a)] Suppose $\int f(\lambda) |\lambda| d\lambda < \infty$. Then, for any $T$ and $n$,
$$
\sqrt{I(X_0^T; Y_0^T)} \leq \frac{\sqrt{2 T \delta_{T, n} \int f(\lambda) |\lambda| d \lambda}+\sqrt{2 T \delta_{T, n} \int f(\lambda) |\lambda| d\lambda + 4 I(X_0^T; Y(\Delta_{T, n}))}}{2}.
$$
\item[(b)]
Suppose $\int f(\lambda) |\lambda| d\lambda < \infty$ and $\int f(\lambda) d\lambda < \infty$. Then, for any $T$ and $n$,
$$
I(X_0^T; Y_0^T) - I(X_0^T; Y(\Delta_{T, n})) \leq  T \sqrt{\delta_{T, n}} \left(\int f(\lambda) |\lambda| d\lambda \right)^{1/2} \left(\int f(\lambda) d\lambda \right)^{1/2}.
$$
Consequently, for any $T$, choosing $n = n(T)$ such that $\lim_{T \to \infty} \delta_{T, n(T)}= 0$, we have
$$
\frac{I(X_0^T; Y_0^T)}{T} - \frac{I(X_0^T; Y(\Delta_{T, n(T)}))}{T} =  O\left(\sqrt{\delta_{T, n(T)}}\right),
$$
as $T$ tends to infinity.
\end{enumerate}
\end{thm}

\begin{proof}

Consider the following parameterized version of the channel (\ref{ct}):
\begin{equation} \label{snr-ct}
Z(t) = \sqrt{snr} \int_0^t X(s) ds + B(t), \quad t \in [0, T],
\end{equation}
where the parameter $snr > 0$ can be regarded as the signal-to-noise ratio of the channel. Obviously, when $snr$ is fixed to be $1$, $Z(t) = Y(t)$ for any $t \in [0, T]$, and moreover, the channel (\ref{snr-ct}) is exactly the same as the channel (\ref{ct}).

Sampling the channel (\ref{snr-ct}) with respect to sampling times $t_{0}, t_{1}, \dots, t_{n}$ that are evenly spaced over $[0, T]$, we obtain the following discrete-time Gaussian channel:
\begin{equation} \label{snr-dt}
Z(t_{i})-Z(t_{i-1})=\sqrt{snr} \int_{t_{i-1}}^{t_{i}} X(s) ds + B(t_{i}) - B(t_{i-1}), \quad i = 1, 2, \dots, n,
\end{equation}
which can be ``normalized'' as follows:
\begin{equation} \label{normalized-snr-dt}
\frac{Z(t_{i})-Z(t_{i-1})}{\sqrt{\delta_{T, n}}} = \sqrt{snr}  \frac{\int_{t_{i-1}}^{t_{i}} X(s) ds}{\sqrt{\delta_{T, n}}} + \frac{B(t_{i}) - B(t_{ i-1})}{\sqrt{\delta_{T, n}}}, \quad i = 1, 2, \dots, n,
\end{equation}
where, at each time $i$, the channel noise $\frac{B(t_{i}) - B(t_{i-1})}{\sqrt{\delta_{T, n}}}$ is a standard Gaussian random variable, $\frac{\int_{t_{ i-1}}^{t_{i}} X(s) ds}{\sqrt{\delta_{T, n}}}$ and $\frac{Z(t_{i})-Z(t_{i-1})}{\sqrt{\delta_{T, n}}}$ should be regarded as the channel input and output, respectively.

By the continuous-time I-MMSE relationship~\cite{gu05} applied to the channel (\ref{snr-ct}), the mutual information of the channel (\ref{ct}) can be computed as
$$
I(X_0^T; Y_0^T) =\frac{1}{2} \int_0^1 \int_0^T \E[\left(X(s)-\E[X(s)|Z_0^T]\right)^2] ds d snr.
$$
And by the discrete-time I-MMSE relationship~\cite{gu05} (or more precisely, its extension~\cite{HanSong2016} to Gaussian memory channels) applied to the channel (\ref{normalized-snr-dt}), we have
$$
I(X_0^T; Y(\Delta_{T, n})) = \frac{1}{2} \int_0^1 \sum_{i=1}^n \E\left[\left(\frac{\int_{t_{i-1}}^{t_{i}} X(s) ds}{\sqrt{\delta_{T, n}}}-\E\left[\left.\frac{\int_{t_{i-1}}^{t_{i}} X(s) ds}{\sqrt{\delta_{T, n}}}\right|Z(\Delta_{T, n})\right]\right)^2\right] d snr.
$$

Obviously, by (\ref{I<=I}), it holds true that $I(X_0^T; Y_0^T) \geq I(X_0^T; Y(\Delta_{T, n}))$. In the following, we will give an upper bound on their difference $I(X_0^T; Y_0^T) - I(X_0^T; Y(\Delta_{T, n}))$, thereby characterizing the closeness between the two quantities. Towards this goal, using the following easily verified that for each $i$,
$$
\E\left[\left(\int_{t_{i-1}}^{t_i} X(s) - \E[X(s)|Z_0^T] \; ds \right)^2\right] \leq \E\left[\left(\int_{t_{i-1}}^{t_i} X(s) - \E[X(s)|Z(\Delta_{T, n})] \; ds \right)^2\right],
$$
we first note that
\begin{align*}
&I(X_0^T; Y_0^T) - I(X_0^T; Y(\Delta_{T, n}))\\
&= \frac{1}{2} \int_0^1 \int_0^T \E[\left(X(s)-\E[X(s)|Z_0^T]\right)^2] ds - \sum_{i=1}^n \E\left[\left(\int_{t_{i-1}}^{t_i} \frac{X(s)-\E[X(s)|Z(\Delta_{T, n})]}{\sqrt{\delta_{T, n}}} ds \right)^2 \right] d snr\\
&\leq \frac{1}{2} \int_0^1 \int_0^T \E[\left(X(s)-\E[X(s)|Z_0^T]\right)^2] ds - \sum_{i=1}^n \E\left[\left(\int_{t_{i-1}}^{t_i} \frac{X(s)-\E[X(s)|Z_0^T]}{\sqrt{\delta_{T, n}}} ds \right)^2 \right] d snr\\
&= \frac{1}{2} \int_0^1 \int_0^T \E[R^2[X(s); Z_0^T]] ds - \sum_{i=1}^n \E\left[\left(\int_{t_{i-1}}^{t_i} \frac{R[X(s); Z_0^T]}{\sqrt{\delta_{T, n}}} ds \right)^2 \right] d snr\\
& = \frac{1}{2} (S_1 + S_2)
\end{align*}
where we have used the shorthand notation $R[X(s); Z_0^T]$ for $X(s)-\E[X(s)|Z_0^T]$ and
{\small $$
S_1 \triangleq \sum_{i=1}^n \int_0^1 \int_{t_{i-1}}^{t_i} \E[R^2[X(s); Z_0^T]] ds  - \sum_{i=1}^n \int_{t_{i-1}}^{t_i} \E\left[R[X(s); Z_0^T] R[X(t_{i-1}); Z_0^T] \right] ds d snr,
$$}
{\small $$
S_2  \triangleq \sum_{i=1}^n \int_0^1 \E\left[\left(\int_{t_{i-1}}^{t_i} R[X(s); Z_0^T] ds \right) R[X(t_{i-1}); Z_0^T] \right] - \sum_{i=1}^n \E\left[\left(\int_{t_{i-1}}^{t_i} \frac{R[X(s); Z_0^T]}{\sqrt{\delta_{T, n}}} ds \right)^2 \right] d snr.
$$}

For the first term, we have
\begin{align}
S_1^2 &= \left(\sum_{i=1}^n \int_{t_{i-1}}^{t_i} \int_0^1 \E[R^2[X(s); Z_0^T]] dsnr ds - \sum_{i=1}^n  \int_{t_{i-1}}^{t_i} \int_0^1 \E\left[R[X(s); Z_0^T] R[X(t_{i-1}); Z_0^T] \right] dsnr ds \right)^2 \nonumber \\
&= \left(\sum_{i=1}^n \int_{t_{i-1}}^{t_i} \int_0^1 \E[R[X(s); Z_0^T] R[X(s)-X(t_{i-1}); Z_0^T]] dsnr ds \right)^2 \nonumber \\
&\leq n \sum_{i=1}^n \left(\int_{t_{i-1}}^{t_i} \int_0^1 \E[R[X(s); Z_0^T] R[X(s)-X(t_{i-1}); Z_0^T]] dsnr ds\right)^2 \nonumber \\
&\leq n \sum_{i=1}^n \int_{t_{i-1}}^{t_i} \int_0^1 \E[R^2[X(s); Z_0^T]] dsnr ds \int_{t_{i-1}}^{t_i} \int_0^1 \E[R^2[X(s)-X(t_{i-1}); Z_0^T]] dsnr ds, \label{r1}
\end{align}
where we have used the Cauchy-Scharz inequality for the last inequality. Now, noticing the fact that
$$
\E[R^2[X(s)-X(t_{i-1}); Z_0^T]] \leq \E[(X(s)-X(t_{i-1}))^2],
$$
we continue as follows:
\begin{align}
S_1^2&\leq n \sum_{i=1}^n \int_{t_{i-1}}^{t_i} \int_0^1 \E[R^2[X(s); Z_0^T]] dsnr ds \int_{t_{i-1}}^{t_i} \int_0^1 \E[(X(s)-X(t_{i-1}))^2] dsnr ds \label{r2}\\
&= n \sum_{i=1}^n \int_{t_{i-1}}^{t_i} \int_0^1 \E[R^2[X(s); Z_0^T]] dsnr ds \int_{t_{i-1}}^{t_i} \int_0^1 \E[X^2(s)+X^2(t_{i-1})-2 X(s) X(t_{i-1})] dsnr ds. \label{r3}
\end{align}
Now, using the fact that, for any $u, v \in \mathbb{R}$,
$$
\E[X(u) X(u+v)] = \int f(\lambda) e^{i v \lambda} d\lambda,
$$
we have
\begin{align}
S_1^2 &\leq n \sum_{i=1}^n \int_{t_{i-1}}^{t_i} \int_0^1 \E[R^2[X(s); Z_0^T]] dsnr ds \int_{t_{i-1}}^{t_i} \int_0^1 \int 2 f(\lambda) |1 - e^{i \lambda (s-t_{i-1})}| d\lambda dsnr ds \label{r4}\\
&\leq n \sum_{i=1}^n \int_{t_{i-1}}^{t_i} \int_0^1 \E[R^2[X(s); Z_0^T]] dsnr ds \int_{t_{i-1}}^{t_i} \int_0^1 \int 2 f(\lambda) |\lambda| (s-t_{i-1}) d\lambda dsnr ds \label{r5} \\
&\leq n \delta_{T, n}^2 \int f(\lambda) |\lambda| d\lambda \int_0^1 \sum_{i=1}^n \int_{t_{i-1}}^{t_i} \E[R^2[X(s); Z_0^T]] ds dsnr \label{r6}\\
&= T \delta_{T, n} \int f(\lambda) |\lambda| d\lambda \int_0^1 \int_0^T \E[(X(s)-\E[X(s)|Z_0^T])^2] ds dsnr \label{r7}\\
&= 2 T \delta_{T, n} \; I(X_0^T; Y_0^T) \int f(\lambda) |\lambda| d\lambda. \label{r8}
\end{align}

For the second term, we have
$$
S_2^2 = \left(\sum_{i=1}^n \int_{t_{i-1}}^{t_i} \int_0^1 \E \left[R[X(s); Z_0^T] \int_{t_{i-1}}^{t_i} \frac{R[X(t_{i-1})-X(u); Z_0^T]}{\delta_{T, n}} du \right] dsnr ds \right)^2.
$$
Starting from this and proceeding in a similar fashion as in (\ref{r1})-(\ref{r8}), we derive
$$
S_2^2 \leq 2 T \delta_{T, n} I(X_0^T; Y_0^T) \int f(\lambda) |\lambda| d\lambda.
$$
Combining the bounds on $S_1$ and $S_2$, we have
\begin{equation} \label{before-two-results}
I(X_0^T; Y_0^T) - I(X_0^T; Y(\Delta_{T, n})) \leq \sqrt{2 T \delta_{T, n}} \left(I(X_0^T; Y_0^T) \int f(\lambda) |\lambda| d\lambda \right)^{1/2}.
\end{equation}
It then immediately follows that
$$
\sqrt{I(X_0^T; Y_0^T)} \leq \frac{\sqrt{2 T \delta_{T, n} \int f(\lambda) |\lambda| d \lambda}+\sqrt{2 T \delta_{T, n} \int f(\lambda) |\lambda| d\lambda + 4 I(X_0^T; Y_{t_0}^{t_n})}}{2},
$$
establishing ($a$). Moreover, together with the fact that
$$
I(X_0^T; Y_0^T) \leq \frac{1}{2} \int_0^T \E[X^2(s)] ds = \frac{T}{2} \int f(\lambda) d\lambda,
$$
the inequality (\ref{before-two-results}) implies that
$$
\left|I(X_0^T; Y_0^T) - I(X_0^T; Y(\Delta_{T, n})) \right| \leq  T \sqrt{\delta_{T, n}} \left(\int f(\lambda) |\lambda| d\lambda \right)^{1/2} \left(\int f(\lambda) d\lambda \right)^{1/2},
$$
establishing ($b$).
\end{proof}

The following corollary, which immediately follows from Theorem~\ref{non-feedback-theorem}, characterizes, among others, how over-sampling approaches the true mutual information of the channel (\ref{ct}) with bandwidth limit.
\begin{co} \label{non-feedback-corollary}
For the continuous-time AWGN channel (\ref{ct}), suppose that the channel input has bandwidth limit $W$ and average power $P$, more precisely, $f(\lambda) = 0$ for any $\lambda \in (-\infty, -W] \cup [W, \infty)$ and $\E[X^2(s)] \leq P$ for any $s \geq 0$. Then, the following two statements hold:
\begin{enumerate}
\item[(a)] For any $T$ and $n$,
$$
\sqrt{I(X_0^T; Y_0^T)} \leq \frac{\sqrt{2 T P W \delta_{T, n}}+\sqrt{2 T P W \delta_{T, n} + 4 I(X_0^T; Y(\Delta_{T, n}))}}{2}.
$$
\item[(b)]
For any $T$ and $n$,
$$
I(X_0^T; Y_0^T) - I(X_0^T; Y(\Delta_{T, n})) \leq  T P \sqrt{W \delta_{T, n}}.
$$
Consequently, for each $T$, choosing $n = n(T)$ such that $\lim_{T \to \infty} \delta_{T, n(T)}= 0$, we have
$$
\frac{I(X_0^T; Y_0^T)}{T} - \frac{I(X_0^T; Y(\Delta_{T, n}))}{T} = O\left(\sqrt{\delta_{T, n(T)}}\right),
$$
as $T$ tends to infinity.
\end{enumerate}
\end{co}

\begin{rem} \label{extension-feedback}
The I-MMSE relationship has played an important role in deriving our results for non-feedback AWGN channels. This powerful tool has been extended~\cite{HanSong2016} to feedback AWGN channels in both discrete and continuous time. A natural question is whether the extended relationship can help us to derive counterpart results to Theorem~\ref{non-feedback-theorem} for the feedback case. As tempting and promising as this idea may look, we failed to find a way to effectively apply the extended I-MMSE relationship, and our treatment for the feedback case, to be detailed in the next section, will use other tools and techniques from the theory of stochastic calculus. Here, we remark that the formulas derived in~\cite{HanSong2016} are valid only with some extra assumption imposed and yet fail to hold true in general. For more detailed explanations and corrected formulas, see Arxiv:1401.3527.
\end{rem}

\section{The Feedback Case} \label{extensions}

In this section, we give quantitative strengthenings of (\ref{limit-version}) for the AWGN channel (\ref{ct}) in the feedback case, which we remind the reader is characterized by the stochastic differential equation in (\ref{SDE}).

The following regularity conditions may be imposed:
\begin{itemize}
\item[(a)] The solution $\{Y(t)\}$ to the stochastic differential equation (\ref{SDE}) uniquely exists.
\item[(b)]
$$
\PX\left(\int_0^T g^2(t, M, Y_0^t) dt < \infty \right)=\PX\left(\int_0^T g^2(t, M, B_0^t) dt < \infty \right)=1.
$$
\item[(c)]
$$
\int_0^T \EX[|g(t, M, Y_0^t)|] dt < \infty.
$$
\item[(d)] \textbf{The uniform Lipschitz condition:} There exists a constant $L > 0$ such that for any $0 \leq s_1, s_2, t_1, t_2 \leq T$, any $Y_0^T, Z_0^T$,
$$
|g(s_1, M, Y_0^{s_2})-g(t_1, M, Z_0^{t_2})| \leq L (|s_1-t_1|+ \|Y_{0}^{s_2}- Z_0^{t_2}\|).
$$

\item[(e)] \textbf{The uniform linear growth condition:} There exists a constant $L > 0$ such that for any $M$ and any $Y_0^T$,
$$
|g(t, M, Y_0^t)| \leq L (1+\|Y_0^t\|).
$$
\end{itemize}

We will need the following lemma, which has already been established~\cite{LiuHan2019} in a slightly more general setting.
\begin{lem} \label{improved-liptser-1}
Assume Conditions (d)-(e). Then, there exists a unique strong solution of (\ref{SDE}) with initial value $Y(0)=0$. Moreover, there exists $\varepsilon > 0$ such that
\begin{equation} \label{exponential-finiteness-1}
\EX [e^{\varepsilon \|Y_0^T\|^2}] < \infty,
\end{equation}
which immediately implies Conditions (b) and (c).
\end{lem}

We will also need the following lemma, whose proof is deferred to Section~\ref{proof-square-exponential}.
\begin{lem} \label{square-exponential}
Let $Z_{max} = \max \{Z_1, Z_2, \dots, Z_n\}$, where the i.i.d. random variables $Z_i \sim N(0, 1/n)$. Then, we have:
\begin{itemize}
\item[a)] For any $0 < \varepsilon < 1$, $\EX[Z_{max}^2] = O(n^{-(1-\varepsilon)})$.
\item[b)] For any $0 < \varepsilon < 1$, $\EX[Z_{max}^4] = O(n^{-(2-\varepsilon)})$.
\item[b)] For any $0 < \varepsilon < 1$, $\EX[e^{Z_{max}^2}] = 1 + O(n^{-(1-\varepsilon)})$.
\end{itemize}
\end{lem}

As detailed below, the following theorem gives the aforementioned quantitative strengthening of (\ref{limit-version}). Here, we mention that with the channel input as in (\ref{SDE}), the discrete-time channel (\ref{dt}) obtained by sampling the channel (\ref{ct}) over $[0, T]$ with respect to $\Delta_{T, n}$ takes the following form:
\begin{equation}  \label{after-sampling-with-feedback}
Y(t_i)=\int_0^{t_i} g(s, M, Y_0^s) ds + B(t_i), \quad i=0, 1, \ldots, n.
\end{equation}
\begin{thm} \label{feedback-theorem}
Fix $T > 0$ and assume Conditions (d)-(e). Then, for any $0 < \varepsilon < 1/2$, we have
$$
I(M; Y_0^T) = I(M; Y(\Delta_{T, n})) + O(\delta_{\Delta_{T, n}}^{1/2-\varepsilon}),
$$
where we recall from Section~\ref{notations} that $Y(\Delta_{T, n}) = \{Y(t_0), Y(t_1), \ldots, Y(t_n)\}$.
\end{thm}

\begin{proof}

First of all, recall from Section~\ref{notations} that for a stochastic process $\{X(s)\}$ and any $t \in \mathbb{R}_+$, we use $\mu_{X_0^t}$ to denote the distribution on $C[0, t]$ induced by $X_0^t$. Throughout the proof, we only have to deal with the case $t=T$, and so we will simply write $\mu_{X_0^T}$ as $\mu_X$. Moreover, we will rewrite $\Delta_{T, n}$ as $\Delta_n$ and $\delta_{T, n}$ as $\delta_n$ for notational simplicity.

Note that an application of Theorem $7.14$ of~\cite{li01} with Conditions ($b$) and ($c$) yields that
\begin{equation} \label{crazy}
\PX\left(\int_0^T \EX^2[g(t, M, Y_0^t)|Y_0^t] dt < \infty \right)=1.
\end{equation}
Then one verifies that the assumptions of Lemma $7.7$ of~\cite{li01} are all satisfied (this lemma is stated under very general assumptions, which are exactly Conditions (b), (c) and (\ref{crazy}) when restricted to our settings), which implies that for any $m$,
\begin{equation}  \label{two-tilde}
\mu_Y \sim \mu_{Y|M=m} \sim \mu_B,
\end{equation}
and moreover, with probability $1$,
{\small \begin{equation}  \label{RN-1}
\frac{d\mu_{Y|M}}{d\mu_B}(Y_0^T|M)=\frac{1}{\EX[e^{\rho_1(M, Y_0^T)}|Y_0^T, M]}, \quad \frac{d\mu_{Y}}{d\mu_B}(Y_0^T)=\frac{1}{\EX[e^{\rho_1(M, Y_0^T)}|Y_0^T]},
\end{equation}}
where
$$
\rho_1(m, Y_0^T) \triangleq -\int_0^T g(s, m, Y_0^s) dY(s)+\frac{1}{2} \int_0^T g(s, m, Y_0^s)^2 ds.
$$
Here we note that $\EX[e^{\rho_1(M, Y_0^T)}|Y_0^T, M]$ is in fact equal to $e^{\rho_1(M, Y_0^T)}$, but we keep it the way it is as above for an easy comparison.

Note that it follows from $\EX[d\mu_B/d\mu_Y(Y_0^T)]=1$ that $\EX[e^{\rho_1(M, Y_0^T)}]=1$, which is equivalent to
\begin{equation}  \label{weak-Novikov}
\EX[e^{-\int_0^T g(s, M, Y_0^s) dB(s)-\frac{1}{2} \int_0^T g(s, M, Y_0^s)^2 ds}]=1.
\end{equation}
Then, a parallel argument as in the proof of Theorem $7.1$ of~\cite{li01} (which requires the condition (\ref{weak-Novikov})) further implies that, for any $n$,
{\small \begin{equation}  \label{RN-2}
\hspace{-1.0cm} \frac{d\mu_{Y(\Delta_n)|M}}{d\mu_{B(\Delta_n)}}(Y(\Delta_n)|M) = \frac{1}{\EX[e^{\rho_1(M, Y_0^T)}|Y(\Delta_n), M]}, \quad \frac{d\mu_{Y(\Delta_n)}}{d\mu_{B(\Delta_n)}}(Y(\Delta_n))=\frac{1}{\EX[e^{\rho_1(M, Y_0^T)}|Y(\Delta_n)]}, ~~\mbox{a.s.}.
\end{equation}}

Next, by the definition of mutual information, we have
\begin{align}
I(M; Y(\Delta_n)) &= \EX\left[\log f_{Y(\Delta_n)|M}(Y(\Delta_n)|M)\right]-\EX\left[\log f_{Y(\Delta_n)}(Y(\Delta_n))\right] \nonumber\\
& = \EX\left[\log \frac{d\mu_{Y(\Delta_n)|M}}{d\mu_{B(\Delta_n)}}(Y(\Delta_n)|M)\right]-\EX\left[\log \frac{d\mu_{Y(\Delta_n)}}{d\mu_{B(\Delta_n)}}(Y(\Delta_n))\right], \label{l-3}
\end{align}
and
\begin{align}
I(M; Y_0^T) & = \EX\left[\log \frac{d \mu_{M, Y_0^T}}{d \mu_{M} \times \mu_{Y_0^T}}(M, Y_0^T)\right] \nonumber\\
& =\EX\left[\log \frac{d\mu_{Y|M}}{d\mu_{B}}(Y_0^T|M)\right]-\EX\left[\log \frac{d\mu_{Y}}{d\mu_B}(Y_0^T)\right], \label{l-4}
\end{align}
where the well-definedness of the Radon-Nikodym derivatives are guaranteed by (\ref{two-tilde}).

By (\ref{l-3}), (\ref{RN-1}) and (\ref{RN-2}), we have
{\small \begin{align*}
I(M; Y(\Delta_n)) &= \EX\left[\log \frac{d\mu_{Y(\Delta_n)|M}}{d\mu_{B(\Delta_n)}}(Y(\Delta_n)|M)\right]-\EX\left[\log \frac{d\mu_{Y(\Delta_n)}}{d\mu_{B(\Delta_n)}}(Y(\Delta_n))\right]\\
&= -\EX[\log \EX[e^{\rho_1(M, Y_0^T)}|Y(\Delta_n), M]] + \EX[\log \EX[e^{\rho_1(M, Y_0^T)}|Y(\Delta_n)]],
\end{align*}}
and
{\small \begin{align*}
I(M; Y_0^T) &= \EX\left[\log \frac{d\mu_{Y|M}}{d\mu_{B}}(Y_0^T|M)\right]-\EX\left[\log \frac{d\mu_{Y}}{d\mu_{B}}(Y_0^T)\right]\\
&= -\EX[\log \EX[e^{\rho_1(M, Y_0^T)}|Y_0^T, M]] + \EX[\log \EX[e^{\rho_1(M, Y_0^T)}|Y_0^T]]\\
&= -\rho_1(M, Y_0^T) + \EX[\log \EX[e^{\rho_1(M, Y_0^T)}|Y_0^T]].
\end{align*}}

We next proceed in the following two steps.

{\bf Step $\bf 1$.} In this step, we show that for any $0 < \varepsilon < 1/2$,
$$
\EX\left[\log \frac{d\mu_{Y|M}}{d\mu_{B}}(Y_0^T|M)\right] - \EX\left[\log \frac{d\mu_{Y(\Delta_n)|M}}{d\mu_{B(\Delta_n)}}(Y(\Delta_n)|M)\right] = O(\delta^{1/2-\varepsilon}_{n}).
$$
Apparently, it suffices to show that for any $0 < \varepsilon < 1/2$,
\begin{equation} \label{Fn-Conv}
\EX[|\log \EX[e^{\rho_1(M, Y_0^T)}|Y(\Delta_n), M] - \rho_1(M, Y_0^T)|] = O(\delta^{1/2-\varepsilon}_{n}).
\end{equation}

Let $\bar{Y}_{\Delta_n, 0}^T$ denote the piecewise linear version of $Y_0^T$ with respect to $\Delta_n$; more precisely, for any $i=0, 1, \dots, n$, let $\bar{Y}_{\Delta_n}(t_{i})=Y(t_{i})$, and for any $t_{i-1} < s < t_{i}$ with $s=\lambda t_{i-1}+(1-\lambda) t_{i}$ for some $0 < \lambda < 1$, let $\bar{Y}_{\Delta_n}(s)=\lambda Y(t_{i-1})+(1-\lambda) Y(t_{i})$. Let $\bar{g}_{\Delta_n}(s, M, \bar{Y}_{\Delta_n, 0}^s)$ denote the piecewise ``flat'' version of $g(s, M, \bar{Y}_{\Delta_n, 0}^s)$ with respect to $\Delta_n$; more precisely, for any $t_{i-1} \leq s < t_{i}$, $\bar{g}_{\Delta_n}(s, M, \bar{Y}_{\Delta_n, 0}^s)=g(t_{i-1}, M, \bar{Y}_{\Delta_n, 0}^{t_{i-1}})$.

Letting
$$
\rho_2(\Delta_n, m, Y_0^T) \triangleq -\int_0^T \bar{g}_{\Delta_n}(s, m, \bar{Y}_{\Delta_n, 0}^s) dY(s)+\frac{1}{2} \int_0^T \bar{g}^2_{\Delta_n}(s, m, \bar{Y}_{\Delta_n, 0}^s) ds,
$$
we have
\begin{align*}
\log \EX[e^{\rho_1(M, Y_0^T)}|Y(\Delta_n), M] & = \log \EX[e^{\rho_2(\Delta_n, M, Y_0^T) +\rho_1(M, Y_0^T)-\rho_2(\Delta_n, M, Y_0^T)}|Y(\Delta_n), M]\\
    & = \log e^{\rho_2(\Delta_n, M, Y_0^T)} \EX[e^{\rho_1(M, Y_0^T)-\rho_2(\Delta_n, M, Y_0^T)}|Y(\Delta_n), M]\\
    & = \rho_2(\Delta_n, M, Y_0^T) +\log \EX[e^{\rho_1(M, Y_0^T)-\rho_2(\Delta_n, M, Y_0^T)}|Y(\Delta_n), M],
\end{align*}
where we have used the fact that
\begin{equation} \label{rho-2-adapted}
\EX[e^{\rho_2(\Delta_n, M, Y_0^T)}|Y(\Delta_n), M]=e^{\rho_2(\Delta_n, M, Y_0^T)},
\end{equation}
since $\rho_2(\Delta_n, M, Y_0^T)$ only depends on $M$ and $Y(\Delta_n)$.

We now prove the following rate of convergence:
\begin{equation} \label{conv-1}
\EX\left[(\rho_1(M, Y_0^T)-\rho_2(\Delta_n, M, Y_0^T))^2\right] = O(\delta_n^{1-\varepsilon}),
\end{equation}
where $0 < \varepsilon < 1$. To this end, we note that
\begin{equation} \label{rho1-rho2}
\rho_1(M, Y_0^T)-\rho_2(\Delta_n, M, Y_0^T) = -\int_0^T (g(s)-\bar{g}_{\Delta_n}(s)) dB(s) - \frac{1}{2} \int_0^T (g(s)-\bar{g}_{\Delta_n}(s))^2 ds,
\end{equation}
where we have rewritten $g(s, M, Y_0^s)$ as $g(s)$, and $\bar{g}_{\Delta_n}(s, M, \bar{Y}_{\Delta_n, 0}^s)$ as $\bar{g}_{\Delta_n}(s)$. It then follows that (\ref{conv-1}) boils down to
\begin{equation} \label{conv-1a}
\EX\left[\left(-\int_0^T (g(s)-\bar{g}_{\Delta_n}(s)) dB(s) - \frac{1}{2} \int_0^T (g(s)-\bar{g}_{\Delta_n}(s))^2 ds \right)^2 \right] = O(\delta_n^{1-\varepsilon}).
\end{equation}
To establish (\ref{conv-1a}), notice that, by the It\^{o} isometry~\cite{ok95}, we have
$$
\EX\left[\left(\int_0^T (g(s)-\bar{g}_{\Delta_n}(s)) dB(s) \right)^2\right] = \EX\left[\int_0^T (g(s)-\bar{g}_{\Delta_n}(s))^2 ds \right].
$$
We next prove that for any $0 < \varepsilon < 1$,
\begin{equation} \label{conv-1b}
\EX\left[\int_0^T (g(s)-\bar{g}_{\Delta_n}(s))^2 ds\right] = O(\delta_n^{1-\varepsilon}).
\end{equation}
To see this, we note that, by Conditions ($d$) and ($e$), there exists $L_1 > 0$ such that for any $s \in [0, T]$ with $t_{i-1} \leq s < t_{i}$,
\begin{align}
\hspace{-1cm} |g(s, M, \bar{Y}_{\Delta_n, 0}^s)-\bar{g}_{\Delta_n}(s, M, \bar{Y}_{\Delta_n, 0}^s)| & = |g(s, M, \bar{Y}_{\Delta_n, 0}^s)-g(t_{i-1}, M, \bar{Y}_{\Delta_n, 0}^{t_{i-1}})| \nonumber\\
& \leq L_1 (|s-t_{i-1}| +\|\bar{Y}_{\Delta_n, 0}^s-\bar{Y}_{\Delta_n, 0}^{t_{i-1}}\|) \nonumber\\
& \leq L_1 (|s-t_{i-1}| + |Y(t_{i})-Y(t_{i-1})|) \nonumber\\
& \leq  L_1 \delta_n+ L_1 \delta_n+L_1 \delta_n \|Y_0^T\|+|B(t_{i})-B(t_{i-1})|. \label{diff-2}
\end{align}
Similarly, there exists $L_2 > 0$ such that for any $s \in [0, T]$,
\begin{align*}
& \hspace{-1cm} |g(s, M, Y_{0}^s)-g(s, M, \bar{Y}_{\Delta_n, 0}^s)| \\
& \leq L_2 \delta_n+L_2 \delta_n \|Y_0^T\|+\max_i \max_{s \in [t_{i-1}, t_i]} \max\{|B(s)-B(t_{i-1})|, |B(s)-B(t_{i})|\}\\
& \leq L_2 \delta_n+L_2 \delta_n \|Y_0^T\| + \max_i \left(\max\{\max_{s \in [t_{i-1}, t_i]} (B(s)-B(t_{i-1})), - \min_{s \in [t_{i-1}, t_i]} (B(s)-B(t_{i-1})), \right.\\
& \hspace{6cm} \left. \max_{s \in [t_{i-1}, t_i]} (B(s)-B(t_{i})), - \min_{s \in [t_{i-1}, t_i]} (B(s)-B(t_{i})) \} \right)\\
& \leq L_2 \delta_n+L_2 \delta_n \|Y_0^T\| + \max_i \left\{\max_{s \in [t_{i-1}, t_i]} (B(s)-B(t_{i-1})) \right\} + \max_i \left\{- \min_{s \in [t_{i-1}, t_i]} (B(s)-B(t_{i-1})) \right\}\\
& \hspace{5cm} + \max_i \left\{\max_{s \in [t_{i-1}, t_i]} (B(s)-B(t_{i})) \right\} + \max_i \left\{- \min_{s \in [t_{i-1}, t_i]} (B(s)-B(t_{i})) \right\}.
\end{align*}
It is well-known (see, e.g., Theorem $2.21$ in~\cite{MortersPeres2010}) that $\max_{s \in [t_{i-1}, t_i]} (B(s)-B(t_{i-1}))$ is distributed as $|B(t_{i})-B(t_{i-1})|$, which, together with the fact that $\{B(t)\}$ has independent increments, leads to
\begin{align*}
& \hspace{-2cm} \E\left[\left(\max_i \left\{\max_{s \in [t_{i-1}, t_i]} (B(s)-B(t_{i-1})) \right\}\right)^2\right] = \E\left[\left(\max_i \left\{|B(t_{i})-B(t_{i-1})| \right\} \right)^2 \right] \\
& \leq \E\left[\left(\max_i \left\{B(t_{i})-B(t_{i-1}) \right\} \right)^2 \right] + \E\left[\left(\min_i \left\{B(t_{i})-B(t_{i-1}) \right\} \right)^2 \right].
\end{align*}
Now, applying Lemma~\ref{square-exponential}~$a)$, we conclude that for any $0 < \varepsilon < 1$,
$$
\E\left[\left(\max_i \left\{\max_{s \in [t_{i-1}, t_i]} (B(s)-B(t_{i-1})) \right\}\right)^2\right] = O(\delta_n^{1-\varepsilon}).
$$
And a completely parallel argument yields that for any $0 < \varepsilon < 1$,
$$
\E\left[\left(\max_i \left\{-\min_{s \in [t_{i-1}, t_i]} (B(s)-B(t_{i-1})) \right\}\right)^2\right] = O(\delta_n^{1-\varepsilon}),
$$
and moreover,
$$
\E\left[\left(\max_i \left\{-\min_{s \in [t_{i-1}, t_i]} (B(s)-B(t_{i})) \right\}\right)^2\right] = O(\delta_n^{1-\varepsilon}),
$$
$$
\E\left[\left(\max_i \left\{-\min_{s \in [t_{i-1}, t_i]} (B(s)-B(t_{i})) \right\}\right)^2\right] = O(\delta_n^{1-\varepsilon}),
$$
where for the latter two, we have used, in addition, the well-known fact that the time reversed Brownian motion is still a Brownian motion. Noticing that, by Lemma~\ref{improved-liptser-1}, $\|Y_0^T\|^2$ is integrable, we arrive at (\ref{conv-1b}), up to an arbitrarily small $\varepsilon > 0$.

A similar argument as above, coupled with Lemma~\ref{square-exponential}~$b)$ (rather than Lemma~\ref{square-exponential}~$a)$), will yield that for any $0 < \varepsilon < 1$,
\begin{equation} \label{conv-1f}
\EX\left[\left(\int_0^T (g(s)-\bar{g}_{\Delta_n}(s))^2 ds \right)^2\right] = O(\delta_n^{2-\varepsilon}),
\end{equation}
which, together with (\ref{conv-1a}) and (\ref{conv-1b}), implies (\ref{conv-1}), as desired.

We now prove the following rate of convergence:
\begin{equation} \label{conv-2}
\EX[\log \EX[e^{\rho_1(M, Y_0^T)-\rho_2(\Delta_n, M, Y_0^T)}|Y(\Delta_n), M]] = O(\delta_n^{1-\varepsilon}),
\end{equation}
where $0 < \varepsilon < 1$. To this end, we first note that
\begin{align*}
\EX[\log \EX[e^{\rho_1(M, Y_0^T)-\rho_2(\Delta_n, M, Y_0^T)}|Y(\Delta_n), M]] & \leq \log \EX[\EX[e^{\rho_1(M, Y_0^T)-\rho_2(\Delta_n, M, Y_0^T)}|Y(\Delta_n), M]] \\
&= \log \EX[e^{\rho_1(M, Y_0^T)-\rho_2(\Delta_n, M, Y_0^T)}]\\
&\stackrel{(a)}{=} \log \EX[e^{-\int_0^T (g(s)-\bar{g}_{\Delta_n}(s)) dB(s) - \frac{1}{2} \int_0^T (g(s)-\bar{g}_{\Delta_n}(s))^2 ds}]\\
&\leq 0,
\end{align*}
where for ($a$), we have used (\ref{rho1-rho2}), and for the last inequality, we have used the fact that
$$
\EX[e^{-\int_0^T (g(s)-\bar{g}_{\Delta_n}(s)) dB(s) - \frac{1}{2} \int_0^T (g(s)-\bar{g}_{\Delta_n}(s))^2 ds}] \leq 1,
$$
which is a well-known fact that follows from Fatou's lemma. For another direction, we have
\begin{align*}
\EX[\log \EX[e^{\rho_1(M, Y_0^T)-\rho_2(\Delta_n, M, Y_0^T)}|Y(\Delta_n), M]] & \geq \EX[\EX[\log e^{\rho_1(M, Y_0^T)-\rho_2(\Delta_n, M, Y_0^T)}|Y(\Delta_n), M]] \\
&= \EX[\EX[\rho_1(M, Y_0^T)-\rho_2(\Delta_n, M, Y_0^T)|Y(\Delta_n), M]] \\
&= \EX[\rho_1(M, Y_0^T)-\rho_2(\Delta_n, M, Y_0^T)]\\
&= - \frac{1}{2} \int_0^T \EX[(g(s)-\bar{g}_{\Delta_n}(s))^2] ds,
\end{align*}
which, together with (\ref{conv-1b}), leads to (\ref{conv-2}).

{\bf Step $\bf 2$.} In this step, we will prove that for any $0 < \varepsilon < 1/2$, there exists a constant $C > 0$ such that for all $n$,
\begin{equation} \label{Gn-Conv}
-\EX\left[\log \frac{d\mu_{Y}}{d\mu_B}(Y_0^T)\right] + \EX\left[\log \frac{d\mu_{Y(\Delta_n)}}{d\mu_{B(\Delta_n)}}(Y(\Delta_n))\right] \leq C \delta_n^{1/2-\varepsilon}.
\end{equation}
First of all, note that by Theorem $6.2.2$ in~\cite{ih93}, we have,
$$
\frac{d \mu_{Y}}{d\mu_B}(Y_0^T) = e^{\int_0^T \hat{g}(Y_0^s) dY(s)-\frac{1}{2} \int_0^T \hat{g}^2(Y_0^s) ds},
$$
where $\hat{g}(Y_0^s)=\EX[g(s, M, Y_0^s)|Y_0^s]$. Moreover, by Theorem $7.23$ of~\cite{li01}, we have,
$$
\frac{d \mu_{Y}}{d\mu_B}(Y_0^T)  = \int \frac{d \mu_{Y|M}}{d\mu_B}(Y_0^T|m) d\mu_M(m),
$$
where
$$
\frac{d \mu_{Y|M}}{d\mu_B}(Y_0^T|m) = e^{\int_0^T g(s, m, Y_0^s) dY(s)-\frac{1}{2} \int_0^T g^2(s, m, Y_0^s) ds}.
$$
Similarly, we have
\begin{align*}
\frac{d\mu_{Y(\Delta_n)}}{d\mu_{B(\Delta_n)}}(Y(\Delta_n)) &= \int \frac{d\mu_{Y(\Delta_n)|M}}{d\mu_{B(\Delta_n)}}(Y(\Delta_n)|m) d\mu_M(m)\\
                                      &= \int \frac{1}{\EX[e^{\rho_1(M, Y_0^T)}|Y(\Delta_n), m]} d\mu_M(m).
\end{align*}
It then follows that
\begin{align*}
\hspace{-1cm} -\EX\left[\log \frac{d\mu_{Y}}{d\mu_B}(Y_0^T)\right] + \EX\left[\log \frac{d\mu_{Y(\Delta_n)}}{d\mu_{B(\Delta_n)}}(Y(\Delta_n))\right] & = -\EX\left[\log e^{-\hat{\rho}_1(Y_0^T)}\right] + \EX \left[ \log \int \frac{1}{\EX[e^{\rho_1(M, Y_0^T)}|Y(\Delta_n), m]} d\mu_M(m) \right]\\
& = \EX \left[ \log \int \frac{e^{\hat{\rho}_1(Y_0^T)-\rho_2(\Delta_n, m, Y_0^T)}}{\EX[e^{\rho_1(M, Y_0^T)-\rho_2(\Delta_n, M, Y_0^T)}|Y(\Delta_n), m]} d\mu_M(m) \right],
\end{align*}
where we have used the shorthand notation $\hat{\rho}_1(Y_0^T)$ for $-\int_0^T \hat{g}(Y_0^s) dY(s)+\frac{1}{2} \int_0^T \hat{g}^2(Y_0^s) ds$, and we have used (\ref{rho-2-adapted}) in deriving the last equality. It then follows that
\begin{align*}
& \hspace{-1cm} \EX \left[ \log \int \frac{e^{\hat{\rho}_1(Y_0^T)-\rho_2(\Delta_n, m, Y_0^T)}}{\EX[e^{\rho_1(M, Y_0^T)-\rho_2(\Delta_n, M, Y_0^T)}|Y(\Delta_n), m]} d\mu_M(m) \right] \\
& \leq \EX \left[ \log \int e^{\hat{\rho}_1(Y_0^T)-\rho_2(\Delta_n, m, Y_0^T)} \EX[e^{-\rho_1(M, Y_0^T)+\rho_2(\Delta_n, M, Y_0^T)}|Y(\Delta_n), m] d\mu_M(m) \right]\\
& \leq \log \EX \left[ e^{\hat{\rho}_1(Y_0^T)-\rho_2(\Delta_n, M, Y_0^T)} \EX[e^{-\rho_1(M, Y_0^T)+\rho_2(\Delta_n, M, Y_0^T)}|Y(\Delta_n), M] \left(\frac{d\mu_Y}{d\mu_B}(Y_0^T)\right)/\left(\frac{d\mu_{Y|M}}{d\mu_B}(Y_0^T|M)\right) \right]\\
& =\log \EX \left[ e^{\rho_1(M, Y_0^T)-\rho_2(\Delta_n, M, Y_0^T)} \EX[e^{-\rho_1(M, Y_0^T)+\rho_2(\Delta_n, M, Y_0^T)}|Y(\Delta_n), M] \right].
\end{align*}
Now, applying the Cauchy-Schwarz inequality, we have
\begin{align*}
& \EX^2 \left[ e^{\rho_1(M, Y_0^T)-\rho_2(\Delta_n, M, Y_0^T)} \EX[e^{-\rho_1(M, Y_0^T)+\rho_2(\Delta_n, M, Y_0^T)}|Y(\Delta_n), M] \right] \\
& \hspace{2cm} \leq \EX \left[ e^{2 \rho_1(M, Y_0^T)- 2 \rho_2(\Delta_n, M, Y_0^T)} \right] \EX\left[\EX^2[e^{-\rho_1(M, Y_0^T)+\rho_2(\Delta_n, M, Y_0^T)}|Y(\Delta_n), M] \right]\\
& \hspace{2cm} \leq \EX \left[ e^{2 \rho_1(M, Y_0^T)- 2 \rho_2(\Delta_n, M, Y_0^T)} \right] \EX\left[\EX [e^{-2\rho_1(M, Y_0^T)+2\rho_2(\Delta_n, M, Y_0^T)}|Y(\Delta_n), M] \right]\\
& \hspace{2cm} = \EX \left[ e^{2 \rho_1(M, Y_0^T)- 2 \rho_2(\Delta_n, M, Y_0^T)} \right] \EX\left[e^{-2\rho_1(M, Y_0^T)+2\rho_2(\Delta_n, M, Y_0^T)} \right].
\end{align*}
Again, applying the Cauchy-Schwarz inequality, we have
\begin{align*}
\EX^2[e^{2 \rho_1(M, Y_0^T)-2 \rho_2(\Delta_n, M, Y_0^T)}] & = \EX^2[e^{-2 \int_0^T (g(s)-\bar{g}_{\Delta_n}(s)) dB(s) - \int_0^T (g(s)-\bar{g}_{\Delta_n}(s))^2 ds}]\\
& = \EX^2[e^{-2 \int_0^T (g(s)-\bar{g}_{\Delta_n}(s)) dB(s) - 4 \int_0^T (g(s)-\bar{g}_{\Delta_n}(s))^2 ds + 3 \int_0^T (g(s)-\bar{g}_{\Delta_n}(s))^2 ds}]\\
& \leq \EX[e^{-4 \int_0^T (g(s)-\bar{g}_{\Delta_n}(s)) dB(s) - 8 \int_0^T (g(s)-\bar{g}_{\Delta_n}(s))^2 ds}] \EX[e^{6 \int_0^T (g(s)-\bar{g}_{\Delta_n}(s))^2 ds}]\\
& \leq \EX[e^{6 \int_0^T (g(s)-\bar{g}_{\Delta_n}(s))^2 ds}],
\end{align*}
where for the last inequality, we have used Fatou's lemma. Now, using a largely parallel argument as in Step $1$ coupled with Lemma~\ref{square-exponential}~$c)$, we conclude that for any $0 < \varepsilon < 1$,
$$
\EX^2[e^{2 \rho_1(M, Y_0^T)-2 \rho_2(\Delta_n, M, Y_0^T)}] \leq \EX[e^{6 \int_0^T (g(s)-\bar{g}_{\Delta_n}(s))^2 ds}] = 1 + O(\delta_n^{1-\varepsilon}),
$$
which further leads to (\ref{Gn-Conv}).

The theorem then immediately follows from Steps $1$, $2$ and the fact that
$$
0 \leq I(M; Y_0^T) - I(M; Y(\Delta_{T, n})).
$$

\end{proof}

\section{Concluding Remarks}

As opposed to the Brownian motion formulation in (\ref{ct}), a continuous-time AWGN channel can be alternatively characterized by the following white noise formulation:
$$
Y(t) = X(t) + Z(t),~~~t \in \mathbb{R}
$$
where $\{Z(t)\}$ is a white Gaussian noise with flat spectral density $1$, and slightly abusing the notation, we have still used $X$ and $Y$, parameterized by $t \in \mathbb{R}$, to represent the channel input and output, respectively. While there is a comprehensive comparison between these two models in~\cite{LiuHan2019}, we emphasize here that the Browninan motion formulation enables a more rigorous information-theoretic examination of AWGN channels and empowers the tools in the theory of stochastic calculus that seem to be essential for more quantitative results for such channels.

Indeed, we believe the framework and techniques developed in this work can be applied to a more quantitative investigation of sampling a wider range of Gaussian channels possibly with different input constraints and related issues, detailed below. First, it is possible that our approach can be applied to obtain a faster rate of convergence for sampling of a peak-power constrained AWGN channel since, as observed in~\cite{Kenneth2018}, the peak-power constraint can provide some much needed uniformity property.
Second, a result by Elias~\cite{Elias1961} has been used to re-derive~\cite{GallagerNakiboglu2010} the somewhat surprising result that for feedback AWGN channels, the Schalkwijk-Kailath scheme yields the decoding error probability that decreases as a second-order exponent in block length. Noticing that Elias' argument in fact used discrete-time MMSE and our treatment essentially quantifies the difference between discrete-time and continuous-time MMSEs, it is worthwhile to investigate whether a rate of convergence result on the decoding error probability can be established for continuous-time feedback AWGN channels. Third, we can also consider sampling of continuous-time additive non-white Gaussian channels. To the best of our knowledge, results in this direction are scarce, but there are a great deal of efforts devoted to discrete-time feedback non-white Gaussian channels (see, e.g.,~\cite{kim10} and references therein). Given the recently obtained counterpart results in discrete time~\cite{LiuTao2019}, it is promising that our treatment can be adapted to such channels, in particular, additive channels with continuous-time auto-regressive and moving average Gaussian noises.

\bigskip \bigskip

{\bf Acknowledgement.} This work is supported by the Research Grants Council of the Hong Kong Special Administrative Region, China, under Project 17301017 and by the National Natural Science Foundation of China, under Project 61871343.

\section*{Appendices} \appendix

\section{Proof of (\ref{Liu-Han-Sampling-Theorem})} \label{already-by-GY}

Consider the continuous-time AWGN channel (\ref{ct}) with the input as in (\ref{non-feedback-case}) or (\ref{feedback-case}). Assuming that the channel input $\{X(t)\}$ is integrable over $[0, T]$, $T > 0$, that is,
\begin{equation} \label{square-integrable}
\int_0^T |X(t)| dt < \infty,~~~a.s.
\end{equation}
we will in this section prove that
\begin{equation} \label{Liu-Han-Sampling-Theorem}
\lim_{n \to \infty} I(M; Y(\Delta_{T, n})) = I(M; Y_0^T).
\end{equation}
As previously mentioned, this result has been implicitly derived in~\cite{GelfandYaglom1957}, and so we only sketch its proof for brevity.

First of all, an appropriately modified version of the proof of Theorem $1.3$ in~\cite{GelfandYaglom1957} can be used to establish that
$$
I(M; Y_0^T) = \sup I(M; Y(\{t_1, t_2, \dots, t_n\})),
$$
where the supremum is over all possible $n \in \mathbb{N}$ and $t_1, t_2, \dots, t_n \in [0, T]$. It immediately follows that
$$
\limsup_{n \to \infty} I(M; Y(\Delta_{T, n})) \leq I(M; Y_0^T).
$$
So, to prove (\ref{Liu-Han-Sampling-Theorem}), it suffices to prove the other direction:
$$
\liminf_{n \to \infty} I(M; Y(\Delta_{T, n})) \geq I(M; Y_0^T).
$$
To this end, for each $m \in \mathbb{N}$, we choose a finite subset $\Pi_{T}^{(m)} \subseteq [0, T]$ such that
$$
I(X_0^T; Y_0^T) = \lim_{m \to \infty} I(X_0^T; Y(\Pi_{T}^{(m)})).
$$
Now, for fixed $\Pi_{T}^{(m)}$ and for each $n \in \mathbb{N}$, we choose $\hat{\Pi}_{T, n}^{(m)} \subset \Delta_{T, n}$ so that $Y(\hat{\Pi}_{T, n}^{(m)})$ is convergent to $Y(\Pi_T^{(m)})$ in distribution. Here we note that the existence of $\{\hat{\Pi}_{T, n}^{(m)}\}$ can be justified by the continuity of $\{Y(t)\}$, which is due to the assumed integrability of $\{X(t)\}$. Then, by Property II at Page $211$ of~\cite{GelfandYaglom1957}, we have
$$
\liminf_{n \to \infty} I(M; Y(\hat{\Pi}_{T, n}^{(m)})) \geq I(M; Y(\Pi_{T}^{(m)})).
$$
Taking $m$ to infinity and using the fact that $\hat{\Pi}_{T, n}^{(m)} \subset \Delta_{T, n}$ for all $n$, we have
$$
\liminf_{n \to \infty} I(M; Y(\Delta_{T, n})) \geq I(M; Y_0^T),
$$
as desired.

\section{Proof of Lemma~\ref{square-exponential}} \label{proof-square-exponential}

We will only prove $a$), the proof of $b$) being largely parallel.

First of all, it can be verified that
\begin{align*}
\EX[Z_{max}^2] & = - \left. t \mathbb{P}(Z_{max}^2 \geq t) \right|_{0}^{\infty} + \int_0^{\infty} \mathbb{P}(Z_{max}^2 \geq t) dt\\
&=- \lim_{t \to \infty} t \mathbb{P}(Z_{max}^2 \geq t) + \int_0^{\infty} \mathbb{P}(Z_{max} \leq -\sqrt{t}) dt + \int_0^{\infty} \mathbb{P}(Z_{max} \geq \sqrt{t}) dt,
\end{align*}
Note that
\begin{align*}
\mathbb{P}(Z_{max} \leq -\sqrt{t}) & = \mathbb{P}(Z_1 \leq -\sqrt{t}, Z_2 \leq -\sqrt{t}, \dots, Z_n \leq -\sqrt{t})\\
&= \mathbb{P}(\sqrt{n} Z_1 \geq \sqrt{n t}, \sqrt{n} Z_2 \geq \sqrt{n t}, \dots, \sqrt{n} Z_n \geq \sqrt{n t})\\
&= \left(\int_{\sqrt{n t}}^{\infty} \frac{1}{\sqrt{2 \pi}} e^{-x^2/2} dx \right)^n,
\end{align*}
where we used the fact that each $\sqrt{n} Z_i$ is a standard normal random variable. Now, using the well-known fact that for any $x > 0$,
$$
\frac{2}{\sqrt{\pi}} \int_x^{\infty} e^{-t^2} dt \leq e^{-x^2},
$$
we derive
\begin{equation} \label{Z-less-than-t}
\mathbb{P}(Z_{max} \leq -\sqrt{t}) = \left(\int_{\sqrt{n t}}^{\infty} \frac{1}{\sqrt{2 \pi}} e^{-x^2/2} dx \right)^n = \left(\frac{1}{2} \frac{2}{\sqrt{\pi}} \int_{\sqrt{\frac{n t}{2}}}^{\infty} e^{-s^2} ds \right)^n \leq \left( \frac{1}{2} \right)^n  e^{-n^2 t/2}.
\end{equation}
And moreover,
\begin{align*}
\mathbb{P}(Z_{max} \geq \sqrt{t}) & = 1 - \mathbb{P}(Z_{max} \leq \sqrt{t})\\
&= 1 - \mathbb{P}(\sqrt{n} Z_1 \leq \sqrt{n t}, \sqrt{n} Z_2 \leq \sqrt{n t}, \dots, \sqrt{n} Z_n \leq \sqrt{n t})\\
&= 1 - \left( \int_{-\infty}^{\sqrt{n t}} \frac{1}{\sqrt{2 \pi}} e^{-x^2/2} dx \right)^n\\
&= 1 - \left(1- \int_{\sqrt{n t}}^{\infty} \frac{1}{\sqrt{2 \pi}} e^{-x^2/2} dx \right)^n.
\end{align*}
Now, using the well-known fact that for any $x > -1$,
$$
\frac{x}{1+x} \leq \log (1+x),
$$
and for any $x \leq 0$,
$$
1+x \leq e^x,
$$
we derive that for any $x$ with $|x| < 1$,
$$
(1-x)^n = e^{n \log (1-x)} \geq e^{-n x /(1-x)} \geq 1 - \frac{n x}{1-x}.
$$
It then follows that
\begin{align}
\mathbb{P}(Z_{max} \geq \sqrt{t}) & = 1 - \left(1- \int_{\sqrt{n t}}^{\infty} \frac{1}{\sqrt{2 \pi}} e^{-x^2/2} dx \right)^n \nonumber\\
&\leq 1 - \left(1- \frac{1}{2} e^{-n t/2} \right)^n \nonumber\\
&\leq 1 - \left(1- \frac{\frac{n}{2} e^{-n t/2}}{1-\frac{1}{2} e^{-n t/2}} \right) \nonumber\\
&= \frac{\frac{n}{2} e^{-n t/2}}{1-\frac{1}{2} e^{-n t/2}}. \label{Z-greater-than-t}
\end{align}
Now, using (\ref{Z-less-than-t}) and (\ref{Z-greater-than-t}), we have that for any $t > 0$,
\begin{align*}
t \mathbb{P}(Z_{max}^2 \geq t) & =  t \mathbb{P}(Z_{max} \leq -\sqrt{t}) + t \mathbb{P}(Z_{max} \geq \sqrt{t})\\
& \leq t \left( \frac{1}{2} \right)^n e^{-n^2 t/2} + t \frac{\frac{n}{2} e^{-n t/2}}{1-\frac{1}{2} e^{-n t/2}},
\end{align*}
which immediately implies that
\begin{equation} \label{later-on}
\lim_{t \to \infty} t \mathbb{P}(Z_{max}^2 \geq t) = 0.
\end{equation}
Moreover, using (\ref{Z-less-than-t}), we derive
\begin{equation} \label{Z-less-than-t-integral}
\int_0^{\infty} \mathbb{P}(Z_{max} \leq -\sqrt{t}) dt \leq \left( \frac{1}{2} \right)^n \int_0^{\infty} e^{-n^2 t/2} dt = \frac{2}{n^2} \left( \frac{1}{2} \right)^n.
\end{equation}
And, using (\ref{Z-greater-than-t}), we have that for any $0 < a < 1$,
\begin{align}
\int_0^{\infty} \mathbb{P}(Z_{max} \geq \sqrt{t}) dt & = \int_0^{\infty} 1 - \left(1- \int_{\sqrt{n t}}^{\infty} \frac{1}{\sqrt{2 \pi}} e^{-x^2/2} dx \right)^n dt \nonumber\\
&\leq \int_0^{n^{-a}} 1 - \left(1- \frac{1}{2} e^{-n t/2} \right)^n dt + \int_{n^{-a}}^{\infty} 1 - \left(1- \frac{1}{2} e^{-n t/2} \right)^n dt \nonumber\\
&\leq \int_0^{n^{-a}} dt + \int_{n^{-a}}^{\infty} 1 - \left(1- \frac{\frac{n}{2} e^{-n t/2}}{1-\frac{1}{2} e^{-n t/2}} \right) dt \nonumber\\
&\leq \int_0^{n^{-a}} dt + \int_{n^{-a}}^{\infty} \frac{\frac{n}{2} e^{-n t/2}}{1-\frac{1}{2} e^{-n t/2}} dt \nonumber\\
&= O(n^{-a}) + O(e^{-n^{1-a}}) \nonumber\\
&= O(n^{-a}). \label{Z-greater-than-t-integral}
\end{align}
Finally, combining (\ref{later-on}), (\ref{Z-less-than-t-integral}) and (\ref{Z-greater-than-t-integral}), we conclude that for any $0 < a < 1$,
$$
\EX[Z_{max}^2] = O(n^{-a}),
$$
as desired.

\end{document}